\documentclass[10pt,conference]{IEEEtran}
\IEEEoverridecommandlockouts

\usepackage{amssymb}
\usepackage[cmex10]{amsmath}
\usepackage{stfloats}
\usepackage{graphicx}
\usepackage{subfigure}
\usepackage{tabularx}
\usepackage{epsfig,epsf,color,balance,cite}
\usepackage{verbatim}
\usepackage{url}
\usepackage{bm}
\usepackage{booktabs}
\usepackage[ruled,linesnumbered]{algorithm2e}

\newtheorem{lemma}{\bf Lemma}
\newtheorem{proof}{Proof}

\usepackage{geometry}
\usepackage{color}
\definecolor{myc1}{rgb}{0,0,0}

\begin{document}

\title{ 
Energy Efficiency Maximization for FAS-Assisted Downlink Communication in Mobile Embodied AI Networks (MEAN) over Interference Channels
}

\author{
\IEEEauthorblockN{
Ruopeng Xu$\IEEEauthorrefmark{1}$,
Zhaohui Yang$\IEEEauthorrefmark{1}$,
Jiaxiang Wang$\IEEEauthorrefmark{2}$,
Chenliang Wu$\IEEEauthorrefmark{1}$,
Zhaoyang Zhang$\IEEEauthorrefmark{1}$
}
	\IEEEauthorblockA{
			$\IEEEauthorrefmark{1}$College of Information Science and Electronic Engineering, Zhejiang University, Hangzhou, China\\
            $\IEEEauthorrefmark{2}$Department of Engineering, King's College London, London, UK\\
          	E-mails:
(ruopengxu, yang\_zhaohui, chenliangwu, ning\_ming)@zju.edu.cn, jiaxiang.wang@kcl.ac.uk, 
		}
\thanks{This work was supported by the National Natural Science Foundation of China (NSFC) under Grants 62394292, 62394290.}
\vspace{-3em}
}
\maketitle

\maketitle

\begin{abstract}
In this paper, we investigate a fluid antenna system (FAS)-assisted downlink mobile embodied AI network (MEAN) over interference channels, where multiple base station (BS)-agent pairs reuse the same spectrum. The BSs employ FASs to improve the communication quality, while the mobile embodied artificial intelligence (AI) agents can adjust their positions according to environment-aware channel information, such as a channel-to-interference-plus-noise map (CINM). Considering both co-channel interference and the energy consumption caused by communication and agent movement, we formulate an energy efficiency (EE) maximization problem by jointly optimizing the agent positions, FAS port selections, and transmit powers. To solve this mixed-integer non-convex problem, we first derive the optimal transmit power in closed form for given agent positions and FAS ports. We then develop an iterative algorithm with adaptive FAS-port optimization and sequential agent-position optimization, together with a low-complexity power-update method. Simulation results demonstrate that the proposed design outperforms the considered benchmark schemes and provides improved feasibility under severe noise conditions.
\end{abstract}

\IEEEpeerreviewmaketitle

\section{Introduction}
The development of artificial intelligence (AI) has booted a series of communication paradigms in sixth generation (6G) communication\cite{xu2026new,11159303,wang2025generative,yang2023energy}. As one of them, mobile embodied AI networks (MEAN) enable intelligent embodied agents to interact with physical environments while communicating with base stations (BSs) for information exchange and task execution\cite{wu2026joint}. Different from conventional communication terminals with fixed locations, embodied agents can actively adjust their positions according to environmental and communication conditions. Such mobility provides additional spatial degrees of freedom for improving wireless communication qualities, especially when channel information at different locations can be obtained from environment-aware databases such as channel knowledge map (CKM)\cite{zeng2021toward} and channel gain map (CGM)\cite{sun2025channel}. However, moving an agent to a position with better communication conditions also introduces additional movement energy consumption. Therefore, an energy-efficient design requires a trade-off between the communication gain obtained from agent movement and its associated energy cost.

Fluid antenna systems (FASs) provide another promising approach to exploiting spatial channel variations\cite{zhang2026finite,11511867,wong2020fluid,11586940}. By dynamically selecting among multiple preset antenna ports, an FAS can adapt its effective channel while requiring only a single RF chain, showing a promising potential to improve communication performance and energy efficiency (EE)\cite{xu2026fluid,xu2026performance}. When FAS is introduced into an MEAN, both port selection and agent mobility can be used to improve communication quality. Such a design is relevant to scenarios, such as a smart manufacturing case, where multiple BSs simultaneously deliver control commands and task data to mobile embodied agents. To improve spectrum utilization, different BS-agent pairs may reuse the same frequency band, which inevitably introduces co-channel interference among receivers. Consequently, each agent receives not only the desired signal from its associated BS but also interference from other BSs. In this case, changing the FAS port of one BS affects not only its desired link but also the interference experienced by other agents. Meanwhile, agent movement further changes both the desired and interference channels. Therefore, agent positions, FAS port selections, and BS transmit powers become coupled.

Existing studies have investigated environment-aware wireless communication using channel information\cite{wu2026joint,zeng2021toward,sun2025channel}, as well as FAS-enabled communication designs\cite{zhang2026finite,11511867,wong2020fluid,11586940,xu2026fluid,xu2026performance}. However, jointly exploiting agent mobility and FAS spatial flexibility in an MEAN over interference channels remains a research gap. Moreover, transmit powers in the communication design are supposed to be coordinated with both agent movement and FAS port selection strategies to satisfy quality of service (QoS) of communications while maintaining high EE, which requires a joint design of agent positions, FAS port selections, and transmit powers. Motivated by the above observations, in this paper, we investigate an FAS-assisted downlink MEAN  with multiple BS-agent pairs sharing the same spectrum. We formulate an EE maximization problem that jointly optimizes the agent-position assignment, FAS port selection, and BS transmit power under the corresponding constraints.

The key contributions of this paper include:
\begin{itemize}
    \item We investigate an FAS-assisted downlink communication in MEAN over interference channels, where multiple BS-agent pairs share the same spectrum. In the considered model, the communication system accounts for the co-channel interference among multiple BS-agent pairs, as well as the additional movement energy consumed by agents when seeking better communication conditions. This establishes a joint movement and communication design framework that captures the trade-off between communication performance improvement and mobility-induced energy consumption.
    \item We formulate an EE maximization problem by jointly optimizing agent positions, FAS port selections, and transmit powers. To solve this problem, we first derive the optimal transmit power vector in closed form for given agent positions and FAS ports. Based on this result, we develop an iterative optimization algorithm that adaptively performs exhaustive or alternating FAS-port optimization according to the size of the port-selection space, and sequentially optimizes the agent positions. In addition, a low-complexity power-update method based on the Sherman–Morrison identity is developed to reduce the computational complexity of position optimization.
    \item We evaluate the proposed FAS-assisted MEAN design through extensive Monte Carlo simulations and compare its EE performance with several benchmark schemes under different numbers of FAS ports, transmission data sizes, and noise power levels. Simulation results demonstrate that the proposed scheme consistently achieves the best EE performance and provides improved feasibility under severe noise conditions, validating the effectiveness of the proposed algorithm.
\end{itemize}

\section{System Model}

As illustrated in Fig.~\ref{SystemModel}, we consider an interference channel with $K$ BS-agent pairs, where each pair consists of one BS equipped with a fluid antenna (FA) and an agent utilizing a fixed-position antenna (FPA) served by the BS. To improve spectrum utilization, we consider a frequency-reuse transmission scenario, in which all BS-agent pairs share the same spectrum band. 
During the transmission, we assume that agents can freely move within the region $\mathcal{D}$, size of $d_1\times d_2$, where channel conditions vary between different positions. The channel-to-interference-plus-noise map (CINM) can provide channel information, similar to the CKM, CGM, etc., plus the information about interference and noise. With the help of CINM, each agent can adjust its spatial position to find better channels to receive signals from its associated BS.

To describe spatial relations, we introduce two-dimensional (2D) Cartesian coordinate system, where the locations of the $k$-th BS and the $k$-th agent are $\mathbf{b}_k = [x_{b_k},y_{b_k}]^T$ and $\mathbf{a}_{k} = [x_{a_k},y_{a_k}]^T$, respectively, and $(\cdot)^T$ is the transpose operation.

\subsection{FAS-Aided Communication Model}
For the $k$-th BS-agent pair, we assume that the $k$-th BS is equipped with a 2D FAS, which includes only one RF chain and $N_k=N_k^1\times N_k^2$ preset ports, where the $N_k^i$ ports are uniformly distributed along a linear space of length $W_k^i\lambda$ for $i\in\{1,2\}$, and take up a grid surface size of $W_k=W_k^1\times W_k^2\lambda^2$ with $\lambda$ being the carrier wavelength. For simplicity of notation, for the $(n_1,n_2)$-th port, we can note it as the $l$-th port, with the mapping relationship $l=(n_1-1)N_k^2+ n_2$.

Based on this, we denote channel between the $l$-th port of FAS belonging the $k$-th BS and the $j$-th agent by
\begin{equation}
   g_{k,j}^l=\sqrt{\alpha_{k,j}}h_{k,j}^l,
\end{equation}
where $h_{k,j}^l $ stands for small-scale fading \cite{11193779}, and $\sqrt{\alpha_{k,j}}$ denotes large-scale fading, which can be defined as\cite{laourine2008capacity}
\begin{equation}
   \alpha_{k,j} = C_0\frac{\psi_{k,j}}{||\mathbf{b}_k-\mathbf{a}_j||_2^\beta}, \forall j,k \in \mathcal{K},
\end{equation}
where $C_0$ is the large-scale reference constant, $\psi_{k,j}$ is an independent and identically distributed (i.i.d.) Gamma random variable, $\beta$ is the path-loss exponent with typical values ranging from $2$ to $6$\cite{yang2015role}, and $\mathcal{K} =\{1,\dots,K\}$. 

We exploit the spatial correlation between FAS ports to model $h_{k,j}^l$, which can be described as
\begin{align}\label{SpatialCorrMatrix}
    \begin{aligned}
      \mathbf{R}_{k,j} = \mathbb{E}[\mathbf{h}_{k,j}\mathbf{h}_{k,j}^H]
      =\begin{bmatrix}
      1 & R_{k,j}^{1,2} & \dots & R_{k,j}^{1,N_k} \\
      R_{k,j}^{2,1} & 1 & \dots & R_{k,j}^{2,N_k} \\
      \vdots & \vdots & &  \vdots\\
      R_{k,j}^{N_k,1} & R_{k,j}^{N_k,2} & \dots & 1 \\
      \end{bmatrix},
    \end{aligned}
\end{align}
where $\mathbf{h}_{k,j} = [h_{k,j}^1,\dots,h_{k,j}^{N_k}]^T$, and entry $R_{k,j}^{x,y}$ is the spatial correlation between the $x$-th port and the $y$-th port of FAS. In particular, the spatial correlation between the $(n_1,n_2)$-th port and $(\tilde{n}_1,\tilde{n}_2)$-th port can be mathematically described as\cite{10303274}
\begin{align}\label{R_jk}
\nonumber
    &R_{k,j}^{(n_1,n_2),(\tilde{n}_1,\tilde{n}_2)}\\
    &=j_0\left(2\pi\sqrt{\left(\frac{|n_1-\tilde{n}_1|}{N_k^1-1}W_k^1\right)^2+\left(\frac{|n_2-\tilde{n}_2|}{N_k^2-1}W_k^2\right)^2}\right),
\end{align}
where $j_0(\cdot)$ is the spherical Bessel function of the first kind. 

\begin{figure}[t]
\centering
\includegraphics[width=1\linewidth]{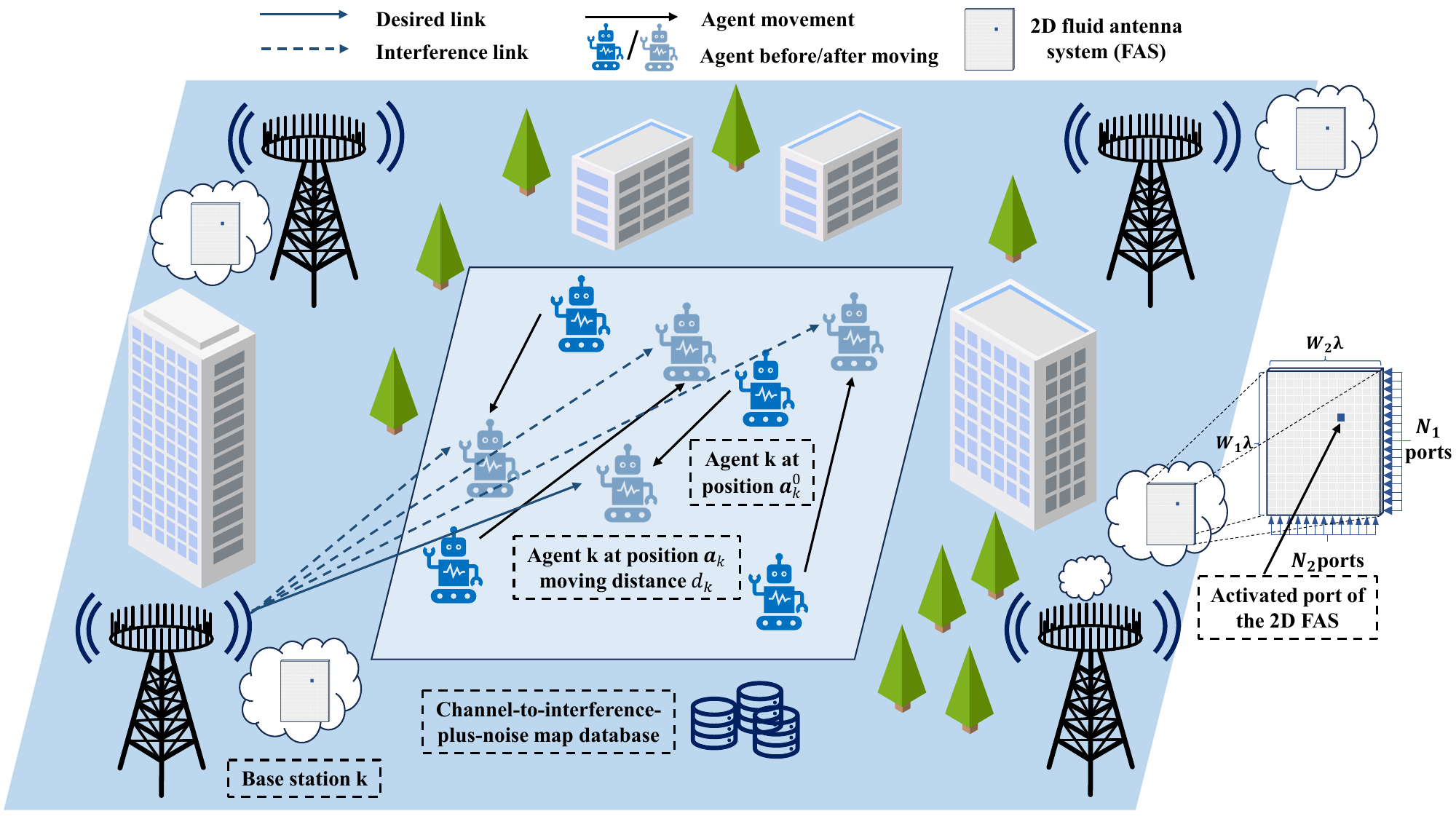}
\caption{An Illustration of the proposed communication system.} 
\label{SystemModel}
\vspace{-0.4cm}
\end{figure}

The aim for the agent $k$ is to receive the target signal from the $k$-th BS, but it also receives interference signals from other BSs. Thus, the received signal at agent $k$ can be given as
\begin{equation}\label{Downlink}
    y_{k} = \sqrt{p_k \alpha_{k,k}} \mathbf{e}_k^H\mathbf{h}_{k,k} s_k + \sum_{j=1, j \neq k}^K \sqrt{p_j \alpha_{j,k}} \mathbf{e}_j^H\mathbf{h}_{j,k} s_j + \eta_k,
\end{equation}
where $\mathbf{e}_k$ is a one-hot vector for BS $k$ to choose the FAS port, $p_k$ is the transmit power, $s_k \sim \mathcal{CN} (0,1)$ is the transmitted signal of BS $k$ following complex Gaussian distribution with a zero mean and a unit variance, and $\eta_k$ $\sim$ $\mathcal{CN}$($0$, $\sigma^2$) is the additive white Gaussian noise (AWGN) in agent $k$. Then, the signal-to-interference-plus-noise ratio ($\mathrm{SINR}$) of the $k$-th agent's received signal at position $\mathbf{a}_k$ can be expressed as
\begin{align}\label{downlink gamma}
    \begin{aligned}
    \gamma_{k} =\frac{p_k  \alpha_{k,k} |\mathbf{e}_k^H\mathbf{h}_{k,k}|^2}{\sum_{j=1, j \neq k}^K  p_j \alpha_{j,k} |\mathbf{e}_j^H\mathbf{h}_{j,k}|^2 + \sigma^2} 
\end{aligned},
\end{align}
based on which, the achievable rate for the $k$-th BS-agent pair can be given as
\begin{equation}
    R_k = B\log_2(1+\gamma_k),
\end{equation}
where $B$ is the bandwidth shared by all BS-agent pairs.

\subsection{Agent Movement Energy Model}

Although each agent can move to find better channels to receive signals, moving from the start point to its destination also introduces extra energy consumption beyond communication. From the perspective of EE, there is a trade-off between the improvement of achievable rate and the extra system energy consumption, both caused by the agent movement. 

To quantify the trade-off, we first model the energy consumption caused by agent movement. If we denote the initial and target positions of agent $k$ as $\mathbf{a}_k^0$ and $\mathbf{a}_k$, respectively, the moving distance of agent $k$ can be denoted by $d_k(\mathbf{a}_k, \mathbf{a}_k^0)$. We can discretize the feasible area $\mathcal{D}$ for agents to move into $M$ candidate points. Specifically, we denote the Cartesian coordinate of the $m$-th position as $\bm{\mu}_m = [x_{\mu_m}, y_{\mu_m}]^T$. Then, the target position of the agent $k$ can be expressed as $\mathbf{a}_k = \sum_{m=1}^M x_{m,k} \bm{\mu}_m$,  where we define $\mathbf{X} = [\mathbf{x}_1,\dots,\mathbf{x}_K]$, and $\mathbf{x}_k = [x_{1,k},\dots,x_{M,k}]^T,\forall k \in \mathcal{K}$.

With an initial position and a chosen target position, agents can invoke tools, such as LLMs, to obtain the route plans. We assume a planar and obstacle-free region\footnote{This assumption allows us to focus on the coupling between agent movement and wireless resource allocation. In environments with obstacles, the Euclidean distance can be replaced by the exact route planning results without altering the overall optimization framework.}, where each agent can move at a constant speed of $v$ directly from its initial location to any candidate position along a feasible straight-line path.
Then, the movement energy consumption can be modeled as\cite{topic2019neural}
\begin{equation}\label{Movement Energy}
    E_k^{\mathrm{move}}= (\tau_r  +\tau_a v^2+ \frac{\tau_c}{v} )d_k(\mathbf{a}_k, \mathbf{a}_k^0)\overset{(a)}{\geq}\tau_0 ||\mathbf{a}_k - \mathbf{a}_k^0||_2 ,
\end{equation}
where $\tau_r$, $\tau_a$, and $\tau_c$ are constant coefficients related to the agent material and the environment\cite{topic2019neural}, $\tau_0 = \tau_r+\frac{3}{2^{2/3}}(\tau_a\tau_c^2)^{1/3}$, and $||\cdot|| $ is the Euclidean norm. The equality $(a)$ holds because $E_k^{\mathrm{move}}$ is a convex function of $v$, and it reaches its minimum value when $v^* = (\frac{\tau_c}{2\tau_a})^{1/3}$.

As a result, we can describe the EE performance after the agent $k$ moves from position $\mathbf{a}_k^0$ to the target position $\mathbf{a}_k$ as 
\begin{align}\label{EE}
    \begin{aligned}
          \mathrm{EE}_k =\frac{B\mathrm{log}_2(1+\gamma_{k})t_{k}}{{\tau_0 ||\sum_{m=1}^M x_{m,k} \bm{\mu}_m - \mathbf{a}_k^0||_2} + p_kt_{k}}, 
    \end{aligned}
\end{align}
where $t_{k} = L_k/R_k$ is the transmission duration, and $L_k$ is the total amount of data to be transmitted.

\subsection{Problem Formulation}
 Then, we can formulate the EE maximization problem as 
\begin{subequations}\label{sys1max0}
    \begin{align} 
       \mathop{\max}_{\mathbf{p},\mathbf{X},\mathbf{E}}
       \nonumber \quad
       &\sum_{k=1}^K \frac{B\mathrm{log}_2(1+\gamma_{k})t_{k}}{{\tau_0 ||\sum_{m=1}^M x_{m,k} \bm{\mu}_m - \mathbf{a}_k^0||_2} + p_kt_{k}},\tag{\ref{sys1max0}}\\
         \textrm{s.t.} \qquad 
         & x_{m, k} \in \{0,1\}, \forall k \in\mathcal{K}, \forall m \in \mathcal{M},\\
         & \sum_{m=1}^M x_{m,k} = 1, \forall k\in \mathcal{K}, \\
         & \sum_{k=1}^K x_{m,k} \leq 1, \forall m \in \mathcal{M}, \\
         & {e}_{n,k} \in \{0,1\},  \forall n \in \mathcal{N}_k, \forall k \in \mathcal{K},\\
         &\sum_{n=1}^{N_k} {e}_{n,k} = 1,\forall k \in \mathcal{K}, \\
         & B\mathrm{log}_2(1+\gamma_{k}) \geq R_k^{\mathrm{min}}, \forall k \in \mathcal{K},\\
         & 0 < p_k \leq P_k^{\mathrm{max}}, \forall k \in \mathcal{K},
    \end{align}
\end{subequations}
where $\mathcal{M} = \{1,\dots,M\}$, $\mathcal{N}_k = \{1,\dots,N_k\}$, $N_\mathrm{max} = \max_{k \in \mathcal{K}}\{N_k\}$, $\mathbf{E} = [\hat{\mathbf{e}}_1,\dots,\hat{\mathbf{e}}_K]$, $\hat{\mathbf{e}}_k = [\mathbf{e}_k^T,\mathbf{0}^T_{N_{\mathrm{max}}-N_k}]^T$, $\mathbf{0}_{a}$ is a zero vector of size $a\times1$, $R_k^{\mathrm{min}}$ is the minimum required achievable rate for the $k$-th BS-agent pair to ensure transmission QoS, and $P_k^{\mathrm{max}}$ is the maximum transmit power for the BS $k$. 

\section{Proposed Algorithm}

\subsection{Optimal Power Control for Fixed Movement and FAS Ports}

Given the agent-position assignment matrix $\mathbf X$, and the FAS port selection matrix $\mathbf E$, problem \eqref{sys1max0} is reduced to a transmit power control problem, which can be reformulated as
\begin{subequations}\label{sub1}
    \begin{align} 
       \mathop{\max}_{\mathbf{p}}
       &\sum_{k=1}^K \frac{L_k B\mathrm{log}_2(1+\frac{p_k |g_{k,k}|^2}{\sum_{j=1, j \neq k}^K  p_j|g_{j,k}|^2 + \sigma^2})}{{\tau_0 d_k B\mathrm{log}_2(1+\frac{p_k |g_{k,k}|^2}{\sum_{j=1, j \neq k}^K p_j |g_{j,k}|^2 + \sigma^2}}) + p_k L_k} ,\tag{\ref{sub1}}\\
         \textrm{s.t.} \ 
         & B\mathrm{log}_2(1+\frac{p_k |g_{k,k}|^2}{\sum_{j=1, j \neq k}^K p_j |g_{j,k}|^2 + \sigma^2} ) \geq R_k^{\mathrm{min}}, \forall k \in \mathcal{K},\\
         & 0 < p_k \leq P_k^{\mathrm{max}}, \forall k \in \mathcal{K},
    \end{align}
\end{subequations}
where $g_{j,k} = \sqrt{\alpha_{j,k}}\mathbf{e}_j^H \mathbf{h}_{j,k}$ is the effective channel from BS $j$ to the agent $k$. 

Based on $(\ref{sub1}\mathrm{a})$, we can obtain the following inequality as
\begin{equation}
p_k \geq \sum_{j=1,j\ne k}^{K}
\frac{\theta_k |g_{j,k}|^2}{|g_{k,k}|^2}p_j
+
\frac{\theta_k\sigma^2}{|g_{k,k}|^2},
\quad \forall k \in \mathcal{K}.
\end{equation}
where $\theta_k = 2^{R_k^{\min}/B}-1$. Furthermore, the above inequalities can be written as
$\mathbf{p}\geq\mathbf{F}\mathbf{p}+\mathbf{u}$,
where the entries in $\mathbf{F}$ hold $F_{k,k}=0$,
$F_{k,j}=\theta_k|g_{j,k}|^2/|g_{k,k}|^2$ for $j\neq k$,
and the entries in $\mathbf{u}$ can be given by $u_k=\theta_k\sigma^2/|g_{k,k}|^2$.


For fixed $\mathbf{X}$ and $\mathbf{E}$, we can easily prove that the objective function of the problem \eqref{sub1} is strictly decreasing with respect to each transmit power component over feasible region. Therefore, the optimal power vector is the minimum feasible power vector satisfying the QoS constraints. and if $\rho(\mathbf F)<1$, where $\rho(\cdot)$ denotes the spectral radius of a matrix, the minimum feasible power vector is given by
\begin{equation}\label{pOpt}
\mathbf p^\star =
(\mathbf I-\mathbf F)^{-1}\mathbf u,
\end{equation}
According to the Perron-Frobenius theorem\cite{120145}, since $\mathbf F$ is entry-wise nonnegative, $\rho(\mathbf F)<1$ guarantees that $\mathbf I-\mathbf F$ is non-singular and that its inverse is also entry-wise nonnegative. Consequently, we can obtain the result shown in \eqref{pOpt}. If $\rho(\mathbf F)\geq 1$ or any entry of $\mathbf p^\star$ exceeds $P_k^{\max}$, the current movement and FAS port selection are considered infeasible.

\subsection{Problem Reformulation and Port Selection Strategies}
With the results shown in \eqref{pOpt}, the original problem \eqref{sys1max0} can be reformulated as
\begin{subequations}\label{sys1max1}
    \begin{align} 
       \mathop{\max}_{\mathbf{X},\mathbf{E}}
       \nonumber \quad
       &\Phi(\mathbf{X},\mathbf{E}) = \sum_{k=1}^K \frac{L_k}{{\tau_0 d_k(\mathbf{X})} + \frac{L_k}{R_k^{\mathrm{min}}}p^\star_k(\mathbf{X},\mathbf{E})},\tag{\ref{sys1max1}}\\
         \textrm{s.t.} \qquad 
        & \rho(\mathbf{F}(\mathbf{X},\mathbf{E})) < 1,\\
        & 0 < p^\star_k(\mathbf{X},\mathbf{E}) \leq P_k^{\mathrm{max}}, \forall k \in \mathcal{K},\\
        \nonumber
        &(\ref{sys1max0}\mathrm{a})-(\ref{sys1max0}\mathrm{e}),
    \end{align}
\end{subequations}
where $d_k(\mathbf{X}) = ||\sum_{m=1}^M x_{m,k} \bm{\mu}_m - \mathbf{a}_k^0||_2$, $p^\star_k(\mathbf{X},\mathbf{E}) = [(\mathbf I-\mathbf F(\mathbf{X},\mathbf{E}))^{-1}\mathbf u(\mathbf{X},\mathbf{E})]_k$ and $[\cdot]_k$ is the $k$-th entry of a vector.

For a given agent-position assignment $\mathbf{X}$, the remaining FAS-port selection problem is a finite discrete optimization problem. Specifically, the total number of possible combinations of port selections is $N_{\mathrm{comb}} = \Pi_{k=1}^K N_k$. We adopt an adaptive search strategy according to the size of this discrete search space. When $N_{\mathrm{comb}} \leq N_{\mathrm{th}}$, where $N_{\mathrm{th}}$ denotes a predefined computational-budget threshold, all possible port selections are exhaustively evaluated, and the globally optimal port selection for a given $\mathbf{X}$ is obtained as
\begin{equation}\label{ESearch}
    \mathbf{E}^\star(\mathbf{X}) = \arg\ \max_{\mathbf{E}} \Phi(\mathbf{X},\mathbf{E}),
\end{equation}
where infeasible combinations that violate the constraint (\ref{sys1max1}a) or the constraint (\ref{sys1max1}b) will be discarded.

When $N_{\mathrm{comb}} > N_{\mathrm{th}}$, an exhaustive search becomes computationally prohibitive. We therefore optimize the FAS-port selections of different BSs in an alternating manner. Let $t=1,2,\dots$ denote the index of the alternating iteration, and let $\mathbf{E}^{(t)}$ denote the port-selection matrix obtained after completing the $t$-th iteration. Given $\mathbf{E}^{(t-1)}$, the BSs are sequentially updated. When updating BS $k$ in the $t$-th iteration, the vectors $\{\mathbf{e}_j^{(t)}\}_{j<k}$ and $\{\mathbf{e}_j^{(t-1)}\}_{j>k}$ are fixed, while $\mathbf{e}_k$ is optimized over all its feasible one-hot configurations. Mathematically, we can update $\mathbf{e}_k^{(t)}$ by solving the following problem as
\begin{subequations}\label{subE}
    \begin{align} 
      \mathbf{e}_k^{(t)} = \arg \mathop{\max}_{\mathbf{e}_k}
       \nonumber 
       &\Phi(\mathbf{X},[\mathbf{e}_1^{(t)},\dots,\mathbf{e}_{k-1}^{(t)},\mathbf{e}_k,\mathbf{e}_{k+1}^{(t-1)},\dots,\mathbf{e}_K^{(t-1)}]) ,\tag{\ref{subE}}\\
         \textrm{s.t.} \qquad 
        \nonumber
        &(\ref{sys1max0}\mathrm{d}),(\ref{sys1max0}\mathrm{e}),(\ref{sys1max1}\mathrm{a}),(\ref{sys1max1}\mathrm{b}).
    \end{align}
\end{subequations}
When port selections do not change between two iterations, i.e., $\mathbf{E}^{(t)} = \mathbf{E}^{(t-1)}$, we find the solution with the given $\mathbf{X}$.

\subsection{Agent Position Optimization for Given Port Selections}
With fixed $\mathbf{E}$, directly enumerating all feasible position assignments requires a combinatorial search over $M!/(M-K)!$ possibilities and is therefore computationally prohibitive for a large $M$. Thus, we optimize the position of each agent sequentially, following a Gauss–Seidel update manner similar to the port optimization above. During the $t$-th iteration, we update the position of the agent $k$ by solving
\begin{subequations}\label{subX}
    \begin{align} 
      \mathbf{x}_k^{(t)} = \arg \mathop{\max}_{\mathbf{x}_k}
       \nonumber 
       &\Phi([\mathbf{x}_1^{(t)},\dots,\mathbf{x}_{k-1}^{(t)},\mathbf{x}_k,\mathbf{x}_{k+1}^{(t-1)},\dots,\mathbf{x}_K^{(t-1)}],\mathbf{E}) ,\tag{\ref{subX}}\\
         \textrm{s.t.} \qquad 
        \nonumber
        &(\ref{sys1max0}\mathrm{a})-(\ref{sys1max0}\mathrm{c}),(\ref{sys1max1}\mathrm{a}),(\ref{sys1max1}\mathrm{b}).
    \end{align}
\end{subequations}
in which an exhaustive search can be performed to find the solution but with a large amount of matrix inversion operations required if directly recomputing \eqref{pOpt} for every candidate position with a large size of $M$. Therefore, we propose the Lemma~\ref{lemma1} to reduce computational cost power control update. 
\begin{lemma}\label{lemma1}
    Let $\mathbf{X}_{k \rightarrow{m}}$ denote the assignment obtained by relocating agent $k$ to candidate position $m$ while keeping all other agents fixed. Correspondingly, we define $\mathbf{F}_{k \rightarrow{m}} = \mathbf{F}(\mathbf{X}_{k \rightarrow{m}},\mathbf{E})$ and $\mathbf{u}_{k \rightarrow{m}} = \mathbf{u}(\mathbf{X}_{k \rightarrow{m}},\mathbf{E})$. Then, the updated power control vector $\mathbf{p}^\star_{k \rightarrow{m}}$ can be calculated by

    \begin{equation}\label{pStarUpdate}
        \mathbf{p}^\star_{k \rightarrow{m}} = \mathbf{p}^\star(\mathbf{X}, \mathbf{E}) + \frac{[\mathbf{Q}^{-1}]_{:,k}(\delta_{k \rightarrow{m}}-\bm{\xi}_{k \rightarrow{m}}^T \mathbf{p}^{\star})}{1+ \bm{\xi}_{k \rightarrow{m}}^T [{\mathbf{Q}}^{-1}]_{:,k}},
    \end{equation}
where $\mathbf{Q}=\mathbf{I}-\mathbf{F}(\mathbf{X},\mathbf{E})$, $\bm{\xi}_{k \rightarrow m}^T = [\mathbf{I}-\mathbf{F}_{k \rightarrow{m}}]_{k,:} - [\mathbf{Q}]_{k,:}$, $\delta_{k \rightarrow{m}} = [\mathbf{u_{k \rightarrow{m}}}]_k - [\mathbf{u}(\mathbf{X},\mathbf{E})]_k$, and $[\cdot]_{k,:}$ and $[\cdot]_{:,k}$ denote the $k$-the row and $k$-th column of a matrix, respectively.

\end{lemma}
\begin{proof}
    See Appendix~\ref{proof1}.
\end{proof}

Based on \eqref{pStarUpdate}, the optimal power response associated with each candidate position can be evaluated without recomputing a matrix inverse. Moreover, for each position $m$, the candidate assignment is feasible if $0 < [\mathbf{p}^\star_{k \rightarrow{m}}]_j \leq P^{\mathrm{max}}_j,\forall j\in \mathcal{K}$. 
Among all feasible candidates, agent $k$ selects the position yielding the largest value of \eqref{subX}.  After all K agents are sequentially updated, one position-optimization sweep is completed. The procedure is repeated until $\mathbf{X}^{(t)} = \mathbf{X}^{(t-1)} $.

\subsection{Overall Algorithm and Complexity Analysis}


\begin{algorithm}[t]
\caption{The Proposed Optimization Algorithm for Solving Problem \eqref{sys1max0}}
\label{alg:overall}
\KwIn{Feasible $\mathbf{p}^{(0)}$, $\mathbf{X}^{(0)}$ and $\mathbf{E}^{(0)}$,
threshold $N_{\mathrm{th}}$}
\KwOut{$\mathbf{X}^{\star}$, $\mathbf{E}^{\star}$, and $\mathbf{p}^{\star}$}

Initialize $r=0$ and obtain $\mathbf{p}^{\star}$ according to
\eqref{pOpt}\;

\Repeat{$\mathbf{X}^{(r)}=\mathbf{X}^{(r-1)}$ $\mathrm{and}$
        $\mathbf{E}^{(r)}=\mathbf{E}^{(r-1)}$}{

    Fix $\mathbf{X}^{(r)}$\;

    \eIf{$N_{\mathrm{comb}}\leq N_{\mathrm{th}}$}{
        Obtain $\mathbf{E}^{(r+1)}$ by the exhaustive search in
        \eqref{ESearch}, with $\mathbf{p}^{\star}$ evaluated by
        \eqref{pOpt}\;
    }{
        Initialize $\mathbf{E}^{(r+1)}=\mathbf{E}^{(r)}$\;

        \Repeat{the FAS-port selection converges}{
            Sequentially update $\mathbf{e}_k$ according to
            \eqref{subE}, and update $\mathbf{p}^{\star}$ by
            \eqref{pOpt}\;
        }
    }

    Fix $\mathbf{E}^{(r+1)}$\;

    Initialize $\mathbf{X}^{(r+1)}=\mathbf{X}^{(r)}$\;

    \Repeat{the position assignment converges}{
        Sequentially update $\mathbf{x}_k$ according to
        \eqref{subX}, and update $\mathbf{p}^{\star}$
        using \eqref{pStarUpdate} in Lemma~\ref{lemma1}\;
    }

    $r\leftarrow r+1$\;
}

\textbf{return}
$\mathbf{X}^{\star}=\mathbf{X}^{(r)}$,
$\mathbf{E}^{\star}=\mathbf{E}^{(r)}$, $\mathrm{and}$ $\mathbf{p}^{\star}$\;
\end{algorithm}

The overall optimization algorithm is summarized in Algorithm~\ref{alg:overall}.
Let $I_E$, $I_X$, and $I_O$ denote the numbers of FAS-port iterations,
position-assignment iterations, and outer iterations required for convergence,
respectively. For each FAS-port selection, evaluating
$\mathbf{p}^{\star}$ through \eqref{pOpt} requires solving a
$K\times K$ linear system, whose complexity is $\mathcal{O}(K^3)$.
Hence, the complexity of the FAS-port optimization is
\begin{equation}
C_E =
\begin{cases}
\mathcal{O}\!\left(N_{\mathrm{comb}}K^3\right),
& N_{\mathrm{comb}}\leq N_{\mathrm{th}},\\[1mm]
\mathcal{O}\!\left(
I_E K^3\sum_{k=1}^{K}N_k
\right),
& N_{\mathrm{comb}}>N_{\mathrm{th}}.
\end{cases}
\end{equation}
For the agent position optimization, by exploiting the results in Lemma~\ref{lemma1}, each candidate can be evaluated using $\mathcal{O}(K)$ vector operations. Considering all $M$ candidate positions for $K$ agents together, the complexity of one position optimization iteration is $C_X=\mathcal{O}\!\left(MK^2\right)$. Therefore, the overall computational complexity of Algorithm~\ref{alg:overall} is $\mathcal{O}\!\left[I_O\left(C_E+I_XMK^2
\right)\right]$.

\section{Simulation Results}

\begin{figure*}[t]
    \centering
            \vspace{-1.5em}
    \subfigure[]{
        \begin{minipage}{0.33\textwidth}
            \centering
            \includegraphics[width=1\textwidth]{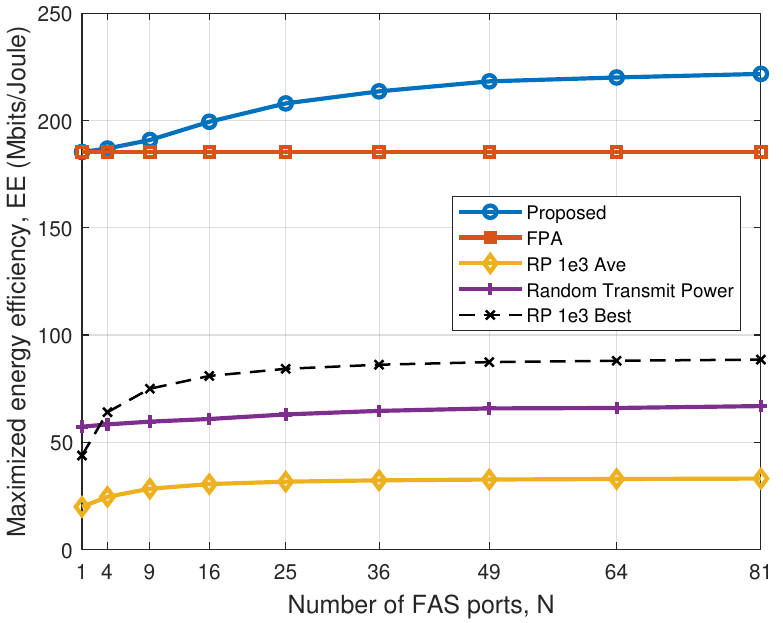}
            \label{N}
    \end{minipage}}
     \vspace{-0.4em}
    \hspace{-3mm}
    \subfigure[]{
        \begin{minipage}{0.33\textwidth}
            \centering
            \includegraphics[width=1\textwidth]{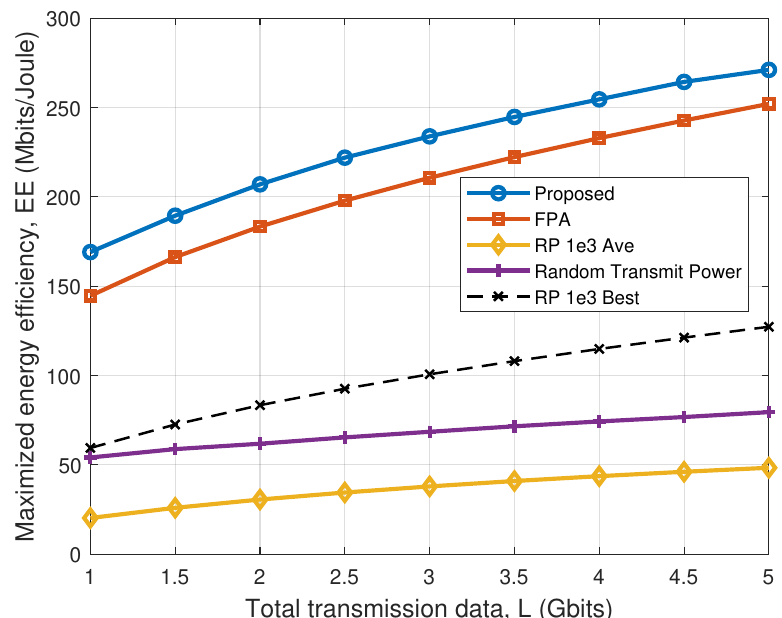}
            \label{dataSize}	
    \end{minipage}}
        \hspace{-5mm}
        \subfigure[]{
            \begin{minipage}{0.33\textwidth}
                \centering
                \includegraphics[width=1\textwidth]{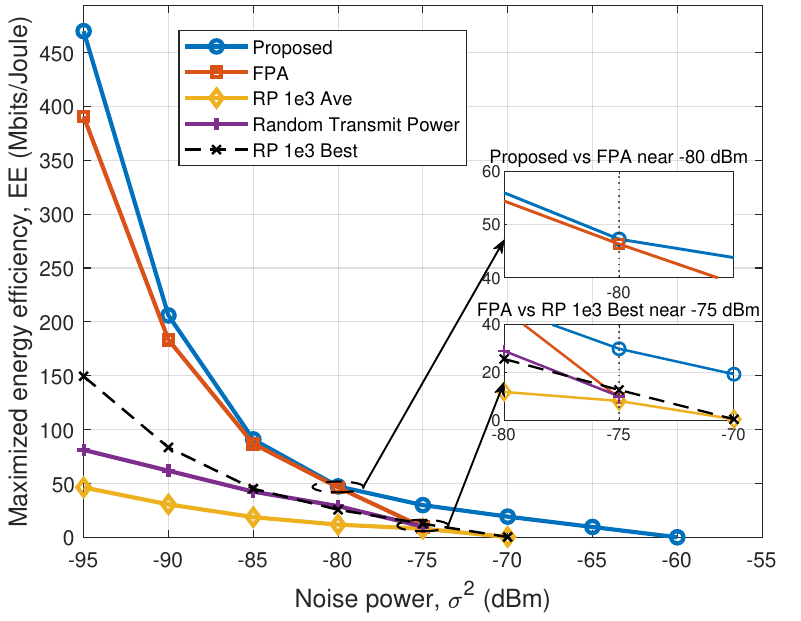}
                \label{NoisePower}	
        \end{minipage}}
            \vspace{-0.4em}
    \caption{Energy efficiency of the system versus the (a) Number of ports $N$, (b) Total transmission data $L$, (c) Noise power $\sigma^2$ .}
    \label{EE Performance} 
    \vspace{-1em}
\end{figure*}

In our simulations, we assume that there are $4$ pairs of BS-agents and $M=200\times200$ candidate positions. Without loss of generality, we assume that the FAS of each BS has the same structure parameters, and specifically, we set $N_1=\cdots=N_K=N=4$ and $W_1=\cdots=W_K=W=1$. Besides, we set $P_1^{\mathrm{max}}=\cdots=P_K^{\mathrm{max}} = 30\ \mathrm{dBm}$, $\sigma^2=-90\ \mathrm{dBm}$, $B=3\ \mathrm{MHz}$, $N_{\mathrm{th}} = 10^5$, and $R_{\mathrm{min}} = 1\ \mathrm{Mbps}$. 
We launch Monte Carlo simulations to evaluate the EE performance of our proposed algorithm, labeled as \textbf{Proposed}, and compare it with benchmark schemes: \textbf{FPA}, which uses an FPA at the BS with optimizing transmit power and agent positions; \textbf{RP 1e3 Ave} and \textbf{RP 1e3 best} both randomly choose the agent positions $1000$ times with optimizing transmit power and FAS port selections, and calculate the average and best EE performance of them, respectively; \textbf{Random Transmit Power}, in which each BS transmits the signals using a random transmit power and optimizes agents positions and FAS port selections.

Fig.~\ref{N} illustrates the EE performance of the system versus the number of ports $N$ of each FAS. When $N=1$, the proposed scheme and FPA achieve nearly the same EE since no port-selection gain is available. As $N$ increases, the EE of the proposed scheme increases, while FPA remains unchanged because it uses a fixed antenna port. The improvement comes from the additional spatial degrees of freedom provided by FAS, which allow more favorable desired-link and interference conditions to be selected. The random-position baselines also benefit from increasing $N$, but remain below the proposed scheme because the agent positions are not optimized. The EE gain gradually saturates for large $N$, as the fixed FAS aperture leads to stronger correlation among increasingly dense ports and hence a diminishing port-selection gain.

Fig.~\ref{dataSize} shows the EE performance of the system versus the total transmission data $L$. The EE generally increases with $L$, and the proposed scheme maintains the highest performance over the considered range. When $L$ is relatively small, the mechanical energy required for movement accounts for a larger portion of the total energy consumption, limiting the benefit of agent movement. As $L$ increases, this fixed movement energy cost is amortized over more transmitted data, and the benefit of jointly optimizing the agent positions, FAS ports, and transmit powers becomes more pronounced. The reduced slope at large $L$ indicates that the EE gradually approaches saturation as the impact of the fixed movement energy on the overall EE decreases with increasing $L$.

Fig.~\ref{NoisePower} demonstrates the EE performance of the system versus the noise power $\sigma^2$. As the noise power increases, a higher transmit power is required to satisfy the minimum-rate constraint, making the problem more likely to become infeasible under the maximum transmit power constraint. The proposed scheme maintains the highest EE over the entire considered range. As highlighted in the zoom-in regions, around $-80\ \mathrm{dBm}$, performance of the proposed algorithm becomes close to that of FPA as the benefit of FAS port selection decreases under increased noise. When the noise power further increases, FPA fails to find feasible solutions beyond $-75\ \mathrm{dBm}$, whereas the benchmark schemes using FAS can still obtain feasible solutions at $-70\ \mathrm{dBm}$, showing that the additional spatial degrees of freedom provided by FAS enlarge the feasible operating range. In comparison, the proposed scheme remains feasible up to $-60\ \mathrm{dBm}$ and becomes infeasible only at higher noise levels, indicating that the proposed algorithm further extends the feasible operating range under severe noise conditions.

\section{Conclusion}

In this paper, we investigated the EE maximization problem for FAS-assisted downlink communication over interference channels in MEAN. By jointly considering the communication energy consumption and agent movement energy, we formulated an EE maximization problem that jointly optimizes the agent positions, FAS port selections, and transmit powers. To solve this problem, we first derived the optimal transmit power for given agent positions and FAS ports, and then developed an iterative optimization algorithm with adaptive FAS port selection and sequential agent position optimization. A low-complexity power update method was further developed to reduce the computational cost during position optimization. Simulation results demonstrated that the proposed scheme achieves higher EE than the considered benchmark schemes, validating the effectiveness of the proposed algorithm.

\begin{appendices}
    \section{Proof of Lemma 1}\label{proof1}
We define $\mathbf{Q}_{k \rightarrow{m}} = \mathbf{I} - \mathbf{F}_{k \rightarrow{m}}$. Then, we have 
\begin{equation}
    \mathbf{Q}_{k \rightarrow{m}} = \mathbf{Q} + \bm{\epsilon}_k \bm{\xi}_{k\rightarrow{m}}^T,
\end{equation}
where $\bm{\epsilon}_k$ denotes the $k$-th standard basis vector in $\mathbb{R}^K$. According to Sherman–Morrison identity\cite{sherman1950adjustment}, if $ \bm{\xi}_{k\rightarrow{m}}^T \mathbf{Q}^{-1} \bm{\epsilon}_k\neq -1$, the inverse of $\mathbf{Q}_{k\rightarrow{m}}$ can be given by
\begin{equation}
    \mathbf{Q}_{k\rightarrow{m}}^{-1} = \mathbf{Q}^{-1} - \frac{\mathbf{Q}^{-1}\bm{\epsilon}_k \bm{\xi}_{k\rightarrow{m}}^T \mathbf{Q}^{-1}}{1+ \bm{\xi}_{k\rightarrow{m}}^T \mathbf{Q}^{-1} \bm{\epsilon}_k}.
\end{equation}

As a result, the updated power control vector after the agent $k$ moving to position $\bm{\mu}_m$ can be calculated as
\begin{align}
\begin{aligned}
\nonumber
       &\mathbf{p}^\star_{k \rightarrow{m}} = \mathbf{Q}_{k\rightarrow{m}}^{-1}\mathbf{u}_{k\rightarrow{m}} \\
       &= \mathbf{Q}^{-1}\mathbf{u} + \delta_{k\rightarrow{m}} \mathbf{Q}^{-1}\bm{\epsilon}_k - \frac{\mathbf{Q}^{-1}\bm{\epsilon}_k \bm{\xi}_{k\rightarrow{m}}^T \mathbf{Q}^{-1}(\mathbf{u}+\delta_{k\rightarrow{m}} \bm{\epsilon}_k)}{1+\bm{\xi}_{k\rightarrow{m}^T \mathbf{Q}^{-1} \bm{\epsilon}_k}} \\
       &= \mathbf{p}^\star + \frac{\mathbf{Q}^{-1}\bm{\epsilon}_k (\delta_{k\rightarrow{m}} - \bm{\xi}_{k\rightarrow{m}^T} \mathbf{Q}^{-1}\mathbf{u})}{1+\bm{\xi}_{k\rightarrow{m}^T \mathbf{Q}^{-1} \bm{\epsilon}_k}}\\
       &= \mathbf{p}^\star + \frac{[\mathbf{Q}^{-1}]_{:,k} (\delta_{k\rightarrow{m}} - \bm{\xi}_{k\rightarrow{m}}^T \mathbf{p}^\star)}{1+\bm{\xi}_{k\rightarrow{m}}^T [\mathbf{Q}^{-1}]_{:,k} },
\end{aligned}
\end{align}
which proves the equality of \eqref{pStarUpdate}.

\end{appendices}

\bibliographystyle{IEEEtran}
\bibliography{MMM}

\end{document}